\documentclass{svjour3}                     % onecolumn (standard format)
\smartqed  % flush right qed marks, e.g. at end of proof
\usepackage{graphicx,amsmath,amssymb}
\newcommand{\ba}{{\mathbf{a}}}
\newcommand{\bb}{{\mathbf{b}}}
\newcommand{\bs}{{\mathbf{s}}}
\newcommand{\bu}{{\mathbf{u}}}
\newcommand{\bv}{{\mathbf{v}}}
\begin{document}

\title{The Linear Complexity of a Class of  Binary Sequences With Optimal Autocorrelation%\thanks{This work was supported by the National Science Foundation of China under Grant (No 61402377), and in part supported by the Open Research Subject of Key Laboratory (Research Base) of Digital Space Security szjj2014-075, and Science and Technology on Communication Security Laboratory Grant 9140C110302150C11004. }
} \subtitle{}

%\titlerunning{Short form of title}        % if too long for running head

\author{Cuiling Fan}

%\authorrunning{Short form of author list} % if too long for running head

\institute{Cuiling Fan \at
School of Mathematics, Southwest Jiaotong University, Chengdu 611756, China. She is also with State Key Laboratory of Information Security
(Institute of Information Engineering, Chinese Academy
of Sciences, Beijing 100093).\\
        \email{cuilingfan@163.com}
}

\date{}%Received: date / Accepted: date}
% The correct dates will be entered by the editor

\maketitle

\begin{abstract}
Binary sequences with optimal autocorrelation and large linear complexity
have important applications in cryptography and communications. Very recently,
a class of binary sequences of period $4p$ with optimal autocorrelation was proposed by
interleaving four suitable Ding-Helleseth-Lam sequences (Des. Codes Cryptogr., DOI
10.1007/s10623-017-0398-5), where $p$ is an odd prime with $p \equiv 1(\bmod 4)$. The objective
of this paper is to determine the minimal polynomial and the linear complexity of
this class of binary optimal sequences via a sequence polynomial approach. It turns out
that this class of sequences has quite good linear complexity.
 \subclass{94A55 \and 94B05}
\end{abstract}

\section{Introduction}

Let $\ba=(a(0),a(1),\cdots,a(N-1))$ be a binary sequence  of period $N$, its periodic autocorrelation is defined by
\begin{eqnarray*}
R_{\ba}(\tau)= \sum\limits_{i=0}^{N-1}(-1)^{a(i)+a(i+\tau)}.
\end{eqnarray*}
Herein and hereafter the addition $i+\tau$ is performed modulo $N$. Let  $\mathbb{Z}_N$ denote the ring of integers modulo $N$. The set
$$
C_{\ba}=\{t\in \mathbb{Z}_N: a(t)=1\}
$$
is called the support of $\ba$, and $\ba$ is called the characteristic sequence of the set $C_{\ba}$. It is easily verified that
\begin{eqnarray}\label{eqn-relation}
R_{\ba}(\tau)= N-4|(C_{\ba}+\tau)\cap C_{\ba}|, \tau\in \mathbb{Z}_N.
\end{eqnarray}

By (\ref{eqn-relation}), one has $R_{\ba}(\tau)\equiv N~(\bmod~4)$ for each $1\leq \tau<N$. Therefore, the optimal value of out-of-phase autocorrelation of binary sequences can be classified into the following four types:
\begin{enumerate}
\item [(A)] $R_{\ba}(\tau)=0$ for $N\equiv 0 \pmod{4}$;
\item [(B)] $R_{\ba}(\tau)\in \{1,-3\}$ for $N\equiv 1 \pmod{4}$;
\item [(C)] $R_{\ba}(\tau)\in \{\pm 2\}$ for $N\equiv 2 \pmod{4}$; and
\item [(D)] $R_{\ba}(\tau)=-1$ for $N\equiv 3 \pmod{4}$.
\end{enumerate}

The sequences in Types (A) and (D) are called perfect sequences and ideal sequences with two-level autocorrelation, respectively. The only known perfect binary sequence up to equivalence is the $(0, 0, 0, 1)$.
It is conjectured that there is no perfect binary sequences of period $N>4$ which
is widely believed to be true in both mathematical and engineer societies.
Hence, it is natural to consider the next smallest value for the out-of-phase autocorrelation of a binary sequence of period $N\equiv 0 \pmod{4}$. That is,
$R_{\ba}(\tau)\in \{0,\pm 4\}$.
If $R_{\ba}(\tau)\in \{0,-4\}$  when $\tau$ ranges from $1$ to $N-1$, then $\ba$ is referred to as a sequence with {\it optimal  autocorrelation value} (with respect to the values) \cite{TG10}.
If $R_{\ba}(\tau)\in \{0,\pm 4\}$ when $\tau$ ranges from $1$ to $N-1$, then  $\ba$ is referred to as a sequence with {\it optimal  autocorrelation magnitude} (with respect to the magnitude of the autocorrelation values)  \cite{Yu2008}.

Binary  sequences with optimal autocorrelation value/magnitude have important applications in many areas of cryptography, communication and radar \cite{GG}.  Finding new binary sequences with optimal autocorrelation value/magnitude  has been an interesting research topic in sequence design. During the last four decades, numerous constructions of binary
sequences with optimal autocorrelation have been reported in the literature (see \cite{Sidelnikov}, \cite{Lempel}, \cite{Ding-H-L}, \cite{Arasu}, \cite{Yu2008}, \cite{CD09}, \cite{TG10} and the references therein).

The linear complexity of a sequence is often defined in
terms of the shortest linear feedback shift register that can generate the sequence.  In order to  resist the well-known Berlekamp-Massey algorithm \cite{Massey},
the employed sequences should  have large linear complexity from the view point of cryptography.  A well-rounded treatment of the linear complexity of sequences with optimal autocorrelation was given in \cite{Wang-Du} and \cite{Li-Tang}.

Very recently,  a new class of binary sequences with optimal autocorrelation magnitude was proposed in \cite{Su-Y-F}. This construction is given as follows. Let $p$ be an odd prime with $p\equiv 1~(\bmod~ 4)$, ${\ba}_0,{\ba}_1,{\ba}_2,{\ba}_3$ be four binary sequences of period $p$ and $\bb=(b(0),b(1),b(2),b(3))$ be a binary sequence of period $4$. Then a binary sequence of period $4p$ can be obtained as below:
\begin{eqnarray}\label{eq-construction}
\bu=I({\ba}_0+b(0), L^d({\ba}_1)+b(1), L^{2d}({\ba}_2)+b(2), L^{3d}({\ba}_2)+b(3)),
\end{eqnarray}
where $I$ and $L$ denote the interleaved operator and the left cyclic shift operator respectively, and $d$ is a positive integer satisfying $4d\equiv 1\pmod{p}$.
It was shown in \cite{Su-Y-F} that the sequence $\bu$ obtained from \eqref{eq-construction} is optimal with respect to the autocorrelation magnitude, i.e., $R_{\bu}(\tau)\in\{0,\pm4\}$ for all $0<\tau<4p$, if the sequences ${\ba}_0,{\ba}_1,{\ba}_2,{\ba}_3$ are chosen to be some Ding-Helleseth-Lam sequences and the sequence $\bb$ satisfies $b(0)=b(2)$ and $b(1)=b(3)$.

The objective of this paper is to determine the minimal polynomial and linear complexity of the  optimal sequences proposed  in \cite{Su-Y-F} based on the sequence polynomial approach.
It turns out that this class of sequences has quite good linear complexity.

\section{Preliminaries}
In this section, we present some basic notation and results on sequences which will be needed in the sequel.

\subsection{Interleaved structure}

Let $\{{\ba}_0,{\ba}_1,\cdots,{\ba}_{T-1}\}$ be a set of $T$ sequences of period $N$. An $N\times T$ matrix $U$ is formed by placing the sequence ${\ba}_i$ on the $i$-th column, where $0\leq i\leq T-1$. Then one can obtain an interleaved sequence $\bu$ of period $NT$ by concatenating the successive rows of the matrix $U$. For simplicity, the interleaved sequence $\bu$ can be written as
\[\bu=I({\ba}_0,{\ba}_1,\cdots,{\ba}_{T-1}),\]
where $I$ denotes the interleaved operator. For more details on interleaved structure, the reader is referred to \cite{GG}.

\subsection{Linear complexity via the sequence polynomial approach}

Let $\bs=(s(i))_{i=0}^{\infty}$ be a sequence over a field $\mathbf{F}$ of period $N$. A polynomial of the form $f(x)=1+c_1x+c_2x^2+\cdots+c_rx^r\in\mathbf{F}[x]$ is called the characteristic polynomial of the sequence $\bs$ if $s(i)=c_1s({i-1})+c_2s({i-2})+\cdots+c_rs({i-r})$ holds for any $i\geq r$, where $\mathbf{F}[x]$ denotes the set of all the polynomials in $x$ over $\mathbf{F}$. The minimal polynomial $\mathbb{M}_{\bs}(x)$ of the sequence $\bs$ is the monic polynomial with the lowest degree in all characteristic polynomials of $\bs$, and the linear complexity of $\bs$ is then defined by the degree of $\mathbb{M}_{\bs}(x)$, that is ${\rm LC}(s)=\deg(\mathbb{M}_{\bs}(x))$. The sequence polynomial of $\bs$, denoted by $\mathbb{P}_{\bs}(x)$, is defined as
$$
\mathbb{P}_{\bs}(x)=\sum_{i=0}^{N-1}s(i)x^i \in \mathbf{F}[x].
$$

There are a few ways to determine the linear span and minimal polynomial of a periodic sequence. One of them is given in the following lemma via the sequence polynomial approach.

\begin{lemma}[\cite{DXS},  p. 87, Theorem 5.3]\label{lem-LC}
Let ${\bs}$ be a sequence over a finite field of period $N$. Then
\begin{enumerate}
  \item [1)]  the minimal polynomial of $\bs$ is $\mathbb{M}_{\bs}(x)=\frac{x^N-1}{\gcd(x^N-1,\mathbb{P}_{\bs}(x))}$; and
  \item [2)]  the linear complexity of $\bs$ is ${\rm LC}(\bs)=N-\deg(\gcd(x^N-1,\mathbb{P}_{\bs}(x))$.
  \end{enumerate}
\end{lemma}

The following result gives relations of the sequence polynomials of some related sequences.

\begin{lemma}[\cite{Wang-Du},\cite{Li-Tang}]\label{lem-Wang}
Let $\ba$ be a binary sequence of period $N$. Then
\begin{enumerate}
  \item [1)]  $\mathbb{P}_{\bb}(x)=x^{N-\tau}\mathbb{P}_{\ba}(x)$ if $\bb=L^{\tau}(\ba)$;
  \item [2)]  $\mathbb{P}_{\bb}(x)=\mathbb{P}_{\ba}(x)+\frac{x^N-1}{x-1}$ if $\bb$ is the complement sequence of $\ba$ (i.e., $b(t)=a(t)+1$ for all $t$); and
  \item [3)]  $\mathbb{P}_{\bu}(x)=\mathbb{P}_{{\ba}_0}(x^4)+x\mathbb{P}_{{\ba}_1}(x^4)+x^2\mathbb{P}_{{\ba}_2}(x^4)+x^3\mathbb{P}_{{\ba}_3}(x^4)$ if $\bu=I({\ba}_0,{\ba}_1,{\ba}_2,{\ba}_3)$.
  \end{enumerate}
 \end{lemma}

Let $\bs$ and $\bv$ be two binary sequences of period $N$. Let $r$ be a positive integer with $\gcd(r,N)=1$.
The sequence $\bv$ is said to be a $r$-decimation of
$\bs$ if $v(t)=s(rt (\bmod N))$ for all $0\leq t<N$. Two sequences
$\bs$ and $\bv$ are said to be equivalent if $\bv$ is a cyclic shift version of the decimation of $\bs$ or its
complement. Otherwise, they are said to be inequivalent.  The following result gives the relationship between the minimal polynomials of two binary equivalent sequences.

\begin{lemma}[Lemma 5, \cite{Wang-Du}]\label{lemma-minimal-relaiton}
Let $\mathbb{M}_{\bs}(x)$ and $\mathbb{M}_{\bv}(x)$ be the minimal polynomials of  two binary sequences
$\bs$  and  $\bv$ of period $N$, respectively. Then we have
\begin{enumerate}
\item[1)] $\mathbb{M}_{\bv}(x)=\mathbb{M}_{\bs}(x)$ if the sequence
$\bv$ can be obtained from $\bs$ by a cyclic shift.

\item[2)] $\deg(\mathbb{M}_{\bv}(x))=\deg(\mathbb{M}_{\bs}(x))$ if the sequence
$\bv$ can be obtained from $\bs$ by a decimation $r$ with $\gcd(r,N)=1$.

\item[3)]If the sequence $\bv$ is a complement of $\bs$, then
\begin{eqnarray*}
\mathbb{M}_{\bv}(x)=\left\{
\begin{array}{ll}
\mathbb{M}_{\bs}(x) (x-1),   &  \mbox{~if~} (x-1)\nmid \mathbb{M}_{\bs}(x),\\
\mathbb{M}_{\bs}(x)/(x-1),   &  \mbox{~if~} (x-1)| \mathbb{M}_{\bs}(x) \mbox{~and~} (x-1)^2\nmid \mathbb{M}_{\bs}(x),\\
\mathbb{M}_{\bs}(x),   &  \mbox{~if~} (x-1)^2| \mathbb{M}_{\bs}(x).
\end{array} \right.
\end{eqnarray*}
\end{enumerate}
\end{lemma}

Lemma \ref{lemma-minimal-relaiton} implies that two binary sequences $\bs$ and $\bv$ of the same period
are inequivalent if $|\deg(\mathbb{M}_{\bs}(x))-\deg(\mathbb{M}_{\bv}(x))|\geq 2$. This fact will be used to judge
when the optimal sequences obtained in \cite{Su-Y-F} (see Theorems \ref{thm-Su-1} and \ref{thm-Su-2} in Section \ref{sec-main}) are inequivalent.

\subsection{Ding-Helleseth-Lam sequences}

Let $p=4f+1$ be an odd prime, where $f$ is a positive integer, and let $\theta$ be a generator of the multiplicative group of the field field $\mathbb{Z}_p$, then the cyclotomic classes $D_i$ of order $4$ are defined as $D_i=\{\theta^{i+4j}: 0\leq j\leq f-1\}$ for $0\leq i\leq 3$. Using the cyclotomic classes of order $4$, Ding, Helleseth and Lam constructed serveral classes of optimal binary sequences with period $p$ as follows.

\begin{theorem}[\cite{Ding-H-L}]\label{thm-Ding}
  Let $p=4f+1=x^2+4y^2$ be an odd prime, where $f, x, y$ are positive integers. Then all the sequences of period $p$  with supports $D_0\cup D_1$, $D_1\cup D_2$, $D_2\cup D_3$, $D_0\cup D_3$ respectively are optimal sequences with autocorrelation values $1$ and $-3$ if and only if $f$ is odd and $y=\pm 1$.
\end{theorem}

%The linear complexity of the Ding-Helleseth-Lam sequences have been determined in \cite{Ding-H-L} via the sequence polynomial approach.
Let $m$ be the order of $2$ modulo $p$ and $\beta$ be a primitive $p$-th root of unity over the finite field $\mathbb{F}_{2^m}$, that is , $\mathbb{F}_{2^m}$ is the splitting field of $x^p-1$. Define
\begin{eqnarray}\label{eq-s(x)}
 \mathbb{S}(x)=\sum_{i\in D_0\cup D_1}x^i, \;\; \mathbb{T}(x)=\sum_{i\in D_1\cup D_2}x^i.
\end{eqnarray}
With the help of the properties of the polynomials $\mathbb{S}(x)$ and $\mathbb{T}(x)$,  the linear complexity of the Ding-Helleseth-Lam sequences was determined in \cite{Ding-H-L}. In the sequel, we need some basic facts about the values of $\mathbb{S}(x)$ and $\mathbb{T}(x)$ at the point $\beta$ used in the proof of Theorem 12 in \cite{Ding-H-L}, which can be easily verified and will play an important role in proving our main results.

\begin{lemma}[\cite{Ding-H-L}]\label{lem-Ding}
With the notation above, we have
$$
\mathbb{S}(\beta^k)=\mathbb{S}(\beta), \mathbb{T}(\beta), \mathbb{S}(\beta)+1, \mathbb{T}(\beta)+1
$$
when $k\in D_0, D_1, D_2, D_3$, respectively.
%\begin{enumerate}
%  \item [1)]  $\mathbb{S}(\beta^k)=\mathbb{S}(\beta), \mathbb{T}(\beta), \mathbb{S}(\beta)+1, \mathbb{T}(\beta)+1$ if $k\in D_0, D_1, D_2, D_3$ respectively;
%  \item [2)]  $(\mathbb{S}(\beta)^2, \mathbb{T}(\beta)^2)=(\mathbb{T}(\beta), \mathbb{S}(\beta)+1)$ if $2\in D_1$;
%  \item [3)] $(\mathbb{S}(\beta)^2, \mathbb{T}(\beta)^2)=(\mathbb{T}(\beta)+1, \mathbb{S}(\beta))$ if $2\in D_3$.
%\end{enumerate}
\end{lemma}

\section{The linear complexity of the optimal sequences obtained from \eqref{eq-construction}}\label{sec-main}

From now on, we adopt the following notation unless otherwise stated:
\begin{itemize}
\item $p=4f+1$ is an odd prime with $f$ being odd.
\item $d$ is a positive integer satisfying $4d\equiv 1\pmod{p}$.
\item $\theta$ is a generator of the multiplicative group of $\mathbb{Z}_p$.
\item $D_i=\{\theta^{i+4j}: 0\leq j\leq f-1\}$ for $0\leq i\leq 3$ are the cyclotomic classes of order $4$.
\item ${\bs}_1,{\bs}_2,{\bs}_2, {\bs}_4$ are the Ding-Helleseth-Lam sequences of period $p$ with the supports $D_0\cup D_1$, $D_0\cup D_3$, $D_1\cup D_2$, $D_2\cup D_3$, respectively.
\item $\mathbb{S}(x)$ and $\mathbb{T}(x)$ are two polynomials given in (\ref{eq-s(x)}).
\end{itemize}

The following result was proved in \cite{Su-Y-F}.

 \begin{theorem}[Theorem 1, \cite{Su-Y-F}]\label{thm-Su-1}
  Let $\bb=(b(0),b(1),b(2),b(3))$ be a binary sequence with $b(0)=b(2)$ and $b(1)=b(3)$, and $({\ba}_0,{\ba}_1,{\ba}_2,{\ba}_3)=({\bs}_3,{\bs}_2,{\bs}_1,{\bs}_1)$. Then the binary sequence $\bu$ constructed from \eqref{eq-construction} is optimal with respect to the autocorrelation magnitude, i.e., $R_{\bu}(\tau)\in\{0,\pm4\}$ for all $0<\tau<4p$.
\end{theorem}

 In what follows, we determine the linear complexity of the optimal sequences in Theorem \ref{thm-Su-1} with the help of Lemmas \ref{lem-LC} and \ref{lem-Ding}. We always assume that $\mathbb{F}_{2^m}$ is the splitting field of $x^p-1$ and $\beta$ is a primitive $p$-th root of $x^p-1$ in $\mathbb{F}_{2^m}$. Then the set $\{\beta^i: i=0,1,2,\cdots, p-1\}$ of roots of $x^p-1$ is a cyclic group of oder $p$ with respect to the multiplication in $\mathbb{F}_{2^m}$. Let $\bu$ be the sequence obtained in Theorem \ref{thm-Su-1}  and $\mathbb{P}_{\bu}(x)$ be its sequence polynomial. It then follows from Lemma \ref{lem-LC} that
  \begin{eqnarray}\label{eq-LCu}
 {\rm LC}(\bu)=4p-\deg(\gcd(x^{4p}-1, \mathbb{P}_{\bu}(x)))=4p-\sum_{i=0}^{p-1}N_i.
\end{eqnarray}
where $N_i=\min\{k_i,4\}$ and $k_i$ denotes the multiplicity of $\beta^i$ as a root of $\mathbb{P}_{\bu}(x)$.

The following lemmas will be needed to prove the main result of this paper in the sequel.

\begin{lemma}\label{lem-polynomial-s-i}
Let symbols be the same as before. Then we have
 \begin{eqnarray*}\label{eq-si(x)}
 \mathbb{P}_{{\bs}_1}(x)={\mathbb{S}(x)}, \; \mathbb{P}_{{\bs}_2}(x)={\mathbb{S}(x^{\theta^3})};  \; \mathbb{P}_{{\bs}_3}(x)={\mathbb{S}(x^{\theta})};\; \mbox{~and~}\mathbb{P}_{{\bs}_4}(x)={\mathbb{S}(x^{\theta^2})}.
\end{eqnarray*}
\end{lemma}
\begin{proof}
The conclusion follows from the definitions of ${\bs}_1,{\bs}_2,{\bs}_3,{\bs}_4$, and the fact that
 \begin{eqnarray*}
\sum_{i\in D_j\cup D_{j+1}}x^i=\sum_{i\in D_0\cup D_1}x^{\theta^j \cdot i}=\mathbb{S}(x^{\theta^j})
\end{eqnarray*}
for  any $0\leq j\leq 3$.
\end{proof}

\begin{lemma}\label{lem-polynomial-u}
For the sequence $\bu$ in Theorem \ref{thm-Su-1}, we have
\begin{eqnarray}\label{eq-su(x)}
 \mathbb{P}_{\bu}(x)&=&\mathbb{S}(x^{4\theta})+x^p\mathbb{S}(x^{4\theta^3})+x^{2p}\mathbb{S}(x^4)+x^{3p}\mathbb{S}(x^4)+\mathbb{P}_{\bb}(x)\cdot\frac{x^{4p}-1}{x^4-1}.
\end{eqnarray}
\end{lemma}
\begin{proof}
Observe that $p=4f+1$ and $4d\equiv 1\pmod{p}$ lead to
$$
d\equiv -f\equiv p-f\equiv 3f+1\pmod{p}
$$
which further implies that
$$
2d\equiv 6f+2 \equiv 2f+1\pmod{p},
$$
and
$$
3d\equiv 9f+3\equiv f+1\pmod{p}.
$$
This together with  \eqref{eq-construction} implies that
\begin{eqnarray*}
\bu&=&I({\bs}_3+b(0), L^{3f+1}({\bs}_2)+b(1), L^{2f+1}({\bs}_1)+b(2), L^{f+1}({\bs}_1)+b(3)).
\end{eqnarray*}
According to Lemmas \ref{lem-polynomial-s-i} and \ref{lem-Wang}, the sequence polynomials of the following sequences
$$
{\bs}_3+b(0), L^{3f+1}({\bs}_2)+b(1), L^{2f+1}({\bs}_1)+b(2), L^{f+1}({\bs}_1)+b(3)
$$
are respectively given by
$$
{\mathbb{S}(x^{\theta})}+b(0)\cdot\frac{x^p-1}{x-1},
$$
$$
{x^{f}\mathbb{S}(x^{\theta^3})}+b(1)\cdot\frac{x^p-1}{x-1},
$$
$$
{x^{2f}\mathbb{S}(x)}+b(2)\cdot\frac{x^p-1}{x-1},
$$
and
$$
{x^{3f}\mathbb{S}(x)}+b(3)\cdot\frac{x^p-1}{x-1}.
$$
It follows from Lemma \ref{lem-Wang} again that
{\small \begin{eqnarray*}
 \mathbb{P}_{\bu}(x)&=&\left({\mathbb{S}(x^{4\theta})}+b(0)\cdot\frac{x^{4p}-1}{x^4-1}\right)+x\cdot\left({x^{4f}\mathbb{S}(x^{4\theta^3})}+b(1)\cdot\frac{x^{4p}-1}{x^4-1}\right)+   \nonumber
         \\ && x^2\cdot\left({x^{8f}\mathbb{S}(x^4)}+b(2)\cdot\frac{x^{4p}-1}{x^4-1}\right)+x^3\cdot\left({x^{12f}\mathbb{S}(x^4)}+b(3)\cdot\frac{x^{4p}-1}{x^4-1}\right)  \nonumber
        \\ &=&\mathbb{S}(x^{4\theta})+x^p\mathbb{S}(x^{4\theta^3})+x^{2p}\mathbb{S}(x^4)+x^{3p}\mathbb{S}(x^4)+ \mathbb{P}_{\bb}(x)\cdot\frac{x^{4p}-1}{x^4-1}.
\end{eqnarray*}}
This completes the proof of this lemma.
\end{proof}

According to \eqref{eq-LCu},  to determine the linear complexity of the sequence $\bu$, it suffices to determine  $N_i$ for each $0\leq i\leq p-1$. This can be done based on  Lemma \ref{lem-polynomial-u}. Specifically, we have the following results.

\begin{lemma}\label{lem-root-1}
$N_i=0$ for each $1\leq i\leq p-1$.
\end{lemma}
\begin{proof}
By \eqref{eq-su(x)} in Lemma \ref{lem-polynomial-u},  $\mathbb{P}_{\bu}(\beta^i)=0$ if and only if $\mathbb{S}(\beta^{4i\theta})+\mathbb{S}(\beta^{4i\theta^3})=0$ due to $\beta^p=1$ and $\beta\ne 1$.  The fact $p=4f+1$ with $f$ being odd implies that $2$ is a non-square element in  $\mathbb{Z}_p$ since the Legendre symbol $(\frac{2}{p})=(-1)^{(p^2-1)/8}=-1$. This means  $2\in D_1\cup D_3$ and then $4\in D_2$. Thus, we have $4\theta\in D_3$ and $4\theta^3\in D_1$. Then, by Lemma \ref{lem-Ding} we have
$$
\mathbb{S}(\beta^{4i\theta})+\mathbb{S}(\beta^{4i\theta^3})=\mathbb{T}(\beta)+\mathbb{T}(\beta)+1=1,
$$
if  $i\in D_0\cup D_2$, and
$$
\mathbb{S}(\beta^{4i\theta})+\mathbb{S}(\beta^{4i\theta^3})=\mathbb{S}(\beta)+\mathbb{S}(\beta)+1=1,
$$
if $i\in D_1\cup D_3$. That is, $\mathbb{S}(\beta^{4i\theta})+\mathbb{S}(\beta^{4i\theta^3})\not=0$ for any $1\leq i\leq p-1$.
This means that $\beta^i$ cannot be a root of  $\mathbb{P}_{\bu}(x)$ for each $1\leq i\leq p-1$, which finishes the proof of this lemma.
\end{proof}

\begin{lemma}\label{lem-root-2}
Let symbols be the same as before. Then we have
\begin{enumerate}
\item [1)] $\gcd(\mathbb{P}_{\bu}(x),x^{4}-1)=x^4-1$ and $N_0=4$ if $\bb=(0,0,0,0)$;
\item [2)] $\gcd(\mathbb{P}_{\bu}(x),x^{4}-1)=x^3+x^2+x+1$ and $N_0=3$ if $\bb=(1,1,1,1)$; and
\item [3)] $\gcd(\mathbb{P}_{\bu}(x),x^{4}-1)=x^2-1$ and $N_0=2$ if $\bb=(1,0,1,0)$ or $\bb=(0,1,0,1)$.
\end{enumerate}
\end{lemma}
\begin{proof}
We only need to
calculate $\gcd(\mathbb{P}_{\bu}(x),x^{4}-1)$, since  $N_0$ is  equal to the degree of the polynomial $\gcd(\mathbb{P}_{\bu}(x),x^{4}-1)$. Let $\mathbb{E}(x)=\mathbb{S}(x^{4\theta})+x^p\mathbb{S}(x^{4\theta^3})+x^{2p}\mathbb{S}(x^4)+x^{3p}\mathbb{S}(x^4)$. It follows from (\ref{eq-s(x)}) that $\mathbb{S}(1)=0$ since $\frac{p-1}{2}$ is even. Thus $(x-1)|\mathbb{S}(x^k)$ and therefore $(x^4-1)|\mathbb{S}(x^{4k})$ for any nonzero integer $k$. It then follows that $(x^4-1)|\mathbb{E}(x)$. This together with
the fact $\gcd(x^4-1,\frac{x^{4p}-1}{x^4-1})=1$ means that
$$
\gcd(\mathbb{P}_{\bu}(x),x^{4}-1)=\gcd(\mathbb{P}_{\bb}(x),x^4-1),
$$
which completes the proof of this lemma.
\end{proof}

Now, we are  in a position to present the main result of this paper.

 \begin{theorem}\label{thm-1}
  Let $\bu$ be the optimal sequence of period $4p$  in Theorem \ref{thm-Su-1}.  Then the minimal polynomial of the sequence $\bu$ is $\mathbb{M}_{\bu}(x)=(x^{4p}-1)/g(x)$ and the linear complexity of $\bu$ is ${\rm LC}(u)=4p-\epsilon$, where
 \begin{enumerate}
\item [1)] $g(x)=x^4-1$ and $\epsilon=4$ if $\bb=(0,0,0,0)$;
\item [2)] $g(x)=x^3+x^2+x+1$ and $\epsilon=3$ if $\bb=(1,1,1,1)$; and
\item [3)] $g(x)=x^2-1$ and $\epsilon=2$ if $\bb=(1,0,1,0)$ or $\bb=(0,1,0,1)$.
\end{enumerate}
\end{theorem}
\begin{proof}
The conclusions follow directly from (\ref{eq-LCu}), and Lemmas \ref{lem-root-1} and \ref{lem-root-2}.
\end{proof}

The following example computed by Magma confirms the results in Theorem \ref{thm-1}.

\begin{example}\label{eg-1}
Let $p=29=4f+1=x^2+4y^2$ for $x=5$, $y=-1$, and $f=7$. Let $\alpha=2$ be a primitive element of $\mathbb{Z}_{p}$.
Then four cyclotomic classes of order $4$ with respect to $\mathbb{Z}_p$ are given by
\begin{eqnarray*}
D_0&=&\{ 1, 7, 16, 20, 23, 24, 25 \},\\
D_1&=&\{ 2, 3, 11, 14, 17, 19, 21 \},\\
D_2&=&\{ 4, 5, 6, 9, 13, 22, 28 \},\\
D_3&=&\{ 8, 10, 12, 15, 18, 26, 27 \}
\end{eqnarray*}
Based on $D_0\cup D_1$, $D_0\cup D_3$,  $D_1\cup D_2$, $D_2\cup D_3$, we generate the following four Ding-Helleseth-Lam sequences
\begin{eqnarray*}
{\bs}_1&=&(0, 1, 1, 1,  0,  0,  0, 1,  0,  0,  0, 1,  0,  0, 1,  0, 1, 1,  0, 1, 1, 1,  0, 1, 1, 1,  0,  0,  0);\\
{\bs}_2&=&(0, 1,  0,  0,  0,  0,  0, 1, 1,  0, 1,  0, 1,  0,  0, 1, 1,  0, 1,  0, 1,  0,  0, 1, 1, 1, 1, 1,  0);\\
{\bs}_3&=&(0,  0, 1, 1, 1, 1, 1,  0,  0, 1,  0, 1,  0, 1, 1,  0,  0, 1,  0, 1,  0, 1, 1,  0,  0,  0,  0,  0, 1);\\
{\bs}_4&=&(0, 0, 0, 0, 1, 1, 1, 0, 1, 1, 1, 0, 1, 1, 0, 1, 0, 0, 1, 0, 0, 0, 1, 0, 0, 0, 1, 1, 1).
\end{eqnarray*}
Take $\bb=(0,0,0,0)$, then $\bu$ in Theorem \ref{thm-Su-1} is the following sequence of period $116$:
\begin{eqnarray*}
&&(0, 0, 0, 0, 0, 1, 1, 0, 1, 1, 1, 0, 1, 1, 0, 1, 1, 1, 1, 0, 1, 1, 1, 0, 1, 0, 1, 1, 0, 0, 0, 0, 0, 1, 1, 1, 1, 0, 1,\\&& 0, 0, 0, 0, 0, 1, 1, 0, 1, 1, 1, 0, 1, 1, 0, 1, 1, 1, 1, 0, 1, 1, 1, 0, 1, 0, 1, 1, 0, 0, 0, 0, 0, 1, 1, 1, 1, 0, 1, \\&& 0, 0, 0, 0, 0, 1, 0, 1, 0, 1, 1, 0, 1, 0, 1, 0, 1, 0, 0, 0, 1, 0, 1, 1, 0, 0, 0, 0, 0, 0, 1, 0, 0, 1, 0, 1, 1).
\end{eqnarray*}
The linear complexity of this sequence is $112$.
Take $\bb=(1,1,1,1)$, then $\bu$ in Theorem \ref{thm-Su-1} is the following sequence of period $116$:
\begin{eqnarray*}
&&(1, 1, 1, 1, 1, 0, 0, 1, 0, 0, 0, 1, 0, 0, 1, 0, 0, 0, 0, 1, 0, 0, 0, 1, 0, 1, 0, 0, 1, 1, 1, 1, 1, 0, 0, 0, 0, 1, 0,\\&&0, 1, 1, 0, 1, 0, 1, 1, 0, 1, 1, 1, 0, 0, 1, 1, 0, 0, 0, 1, 1, 1, 0, 0, 0, 1, 1, 0, 0, 0, 0, 0, 0, 1, 1, 1, 1, 0, 0, 1, \\&& 1, 1, 1, 1, 1, 0, 1, 0, 1, 0, 0, 1, 0, 1, 0, 1, 0, 1, 1, 1, 0, 1, 0, 0, 1, 1, 1, 1, 1, 1, 0, 1, 1, 0, 1, 0, 0).
\end{eqnarray*}
The linear complexity of this sequence is $113$. Take $\bb=(1,0,1,0)$ we get the following sequence of period $116$:
\begin{eqnarray*}
&&(1, 0, 1, 0, 1, 1, 0, 0, 0, 1, 0, 0, 0, 1, 1, 1, 0, 1, 0, 0, 0, 1, 0, 0, 0, 0, 0, 1, 1, 0, 1, 0, 1, 1, 0, 1, 0, 0, 0,\\&&1, 1, 0, 0, 0, 0, 0, 1, 1, 1, 0, 1, 1, 0, 0, 1, 1, 0, 1, 1, 0, 1, 1, 0, 1, 1, 0, 0, 1, 0, 1, 0, 1, 1, 0, 1, 0, 0, 1, 1, \\&& 0, 1, 0, 1, 0, 0, 0, 0, 0, 0, 1, 1, 1, 1, 1, 1, 1, 1, 0, 1, 1, 1, 1, 0, 0, 1, 0, 1, 0, 1, 1, 1, 0, 0, 0, 0, 1).
\end{eqnarray*}
The linear complexity of this sequence is $114$.
\end{example}

With the same method used in Theorem \ref{thm-Su-1}, more optimal binary sequences of period $4p$ were obtained in \cite{Su-Y-F} from the Ding-Helleseth-Lam sequences.

 \begin{theorem}[Theorem 2, \cite{Su-Y-F}]\label{thm-Su-2}
 Let $\bb=(b(0),b(1),b(2),b(3))$ be a binary sequence with $b(0)=b(2)$ and $b(1)=b(3)$, and $({\ba}_0,{\ba}_1,{\ba}_2,{\ba}_3)$ be chosen from
{\small
 \begin{eqnarray}\label{eq-si}
 \{({\bs}_2,{\bs}_3,{\bs}_1,{\bs}_1),({\bs}_4,{\bs}_1,{\bs}_2,{\bs}_2),({\bs}_1,{\bs}_4,{\bs}_2,{\bs}_2),({\bs}_4,{\bs}_1,{\bs}_3,{\bs}_3),({\bs}_1,{\bs}_4,{\bs}_3,{\bs}_3),({\bs}_2,{\bs}_3,{\bs}_4,{\bs}_4)\}.
\end{eqnarray}
}
Then the binary sequence $\bu$ constructed from \eqref{eq-construction} is optimal with respect to the autocorrelation magnitude.
 \end{theorem}

Similar to the proof of Theorem \ref{thm-1}, we can also determine the minimal polynomial and linear complexity of the sequence  in Theorem \ref{thm-Su-2}. The details are left to the reader.

 \begin{theorem}\label{thm-2}
  Let $\bu$ be the optimal sequence of period $4p$ obtained in Theorem \ref{thm-Su-2}, where $b(0)=b(2)$ and $b(1)=b(3)$, and $({\ba}_0,{\ba}_1,{\ba}_2,{\ba}_3)$ be any element chosen from \eqref{eq-si}. Then the minimal polynomial of the sequence $\bu$ is $\mathbb{M}_{\bu}(x)=(x^{4p}-1)/g(x)$ and the linear complexity of $\bu$ is ${\rm LC}(u)=4p-\epsilon$, where
 \begin{enumerate}
\item [1)] $g(x)=x^4-1$ and $\epsilon=4$ if $\bb=(0,0,0,0)$;
\item [2)] $g(x)=x^3+x^2+x+1$ and $\epsilon=3$ if $\bb=(1,1,1,1)$; and
\item [3)] $g(x)=x^2-1$ and $\epsilon=2$ if $\bb=(1,0,1,0)$ or $\bb=(0,1,0,1)$.
\end{enumerate}
\end{theorem}

\begin{example}\label{eg-2}
Let ${\bs}_1,{\bs}_2,{\bs}_3, {\bs}_4$ be the sequences in Example \ref{eg-1}. Take $({\ba}_0,{\ba}_1,{\ba}_2,{\ba}_3)=({\bs}_4,{\bs}_1,{\bs}_2,{\bs}_2)$ and $\bb=(0,1,0,1)$, then $\bu$ in Theorem \ref{thm-Su-2} is the following sequence of period $116$:
\begin{eqnarray*}
&&(0, 1, 1, 0, 0, 0, 1, 1, 0, 0, 0, 0, 0, 0, 1, 1, 1, 1, 0, 0, 1, 1, 1, 1, 1, 1, 0, 1, 0, 1, 0, 0, 1, 0, 1, 0, 1, 0, 1,\\&&1, 1, 0, 1, 0, 0, 1, 1, 1, 1, 1, 1, 0, 1, 1, 0, 1, 0, 0, 0, 1, 1, 1, 1, 0, 0, 1, 0, 0, 0, 1, 0, 0, 1, 0, 0, 0, 0, 1, 0, \\&& 0, 0, 1, 0, 1, 0, 0, 1, 1, 1, 1, 1, 0, 0, 0, 0, 1, 0, 0, 1, 1, 0, 1, 0, 1, 1, 0, 1, 1, 1, 0, 0, 1, 1, 0, 0, 0).
\end{eqnarray*}
The linear complexity of this sequence is $114$, which confirms the result in Theorem \ref{thm-2}.
\end{example}

\begin{remark}\label{remark-2}

Let $\bu$ and $\bu'$ be any two sequences in Theorem \ref{thm-Su-1} or \ref{thm-Su-2} which are respectively written as
$$
\bu=I({\ba}_0+b(0), L^d({\ba}_1)+b(1), L^{2d}({\ba}_2)+b(2), L^{3d}({\ba}_2)+b(3)),
$$
and
$$
\bu'=I({\ba}'_0+b'(0), L^d({\ba}'_1)+b'(1), L^{2d}({\ba'}_2)+b(2), L^{3d}({\ba'}_2)+b'(3)),
$$
where $({\ba}_0,{\ba}_1,{\ba}_2,{\ba}_3)$, $({\ba}'_0,{\ba}'_1,{\ba}'_2,{\ba}'_3)$, $\bb$
and $\bb'$ satisfy the conditions in  Theorem \ref{thm-Su-1} or \ref{thm-Su-2}. Then a natural question one would ask is when
$\bu$ and $\bu'$ are inequivalent (or equivalent). It may be difficult to answer this question in general. However, this can be done in the following cases.

Case 1), when $({\ba}_0,{\ba}_1,{\ba}_2,{\ba}_3)=({\ba}'_0,{\ba}'_1,{\ba}'_2,{\ba}'_3)$ and $\bb$ is the complement of $\bb'$: In this case, it is obvious that $\bu$ and $\bu'$ are equivalent since $\bu$ is the complement of $\bu'$.

Case 2), when $\bb\neq \bb'$ and $\bb$ is not the complement of $\bb'$:  In this case, $\bu$ and $\bu'$ are inequivalent due to Lemma \ref{lemma-minimal-relaiton}, and Theorems \ref{thm-1} and \ref{thm-2}.

\end{remark}

\section{Concluding Remarks}

In this paper, the minimal polynomial and linear complexity of the optimal sequences with period $4p$  from interleaving four Ding-Helleseth-Lam sequences of period $p$ were completely determined via the sequence polynomial approach. It turns out this class of binary optimal sequences have very large linear complexity.
It would be interesting to construct more binary sequences with optimal autocorrelation and large linear complexity.

\section*{Acknowledgments}

The author is very grateful to the reviewers and the Editor for their valuable comments that improved the presentation and quality of this paper. This work was supported by
the Natural Science Foundation of China under Grants 11571285 and 61661146003, and the Sichuan Provincial Youth Science and Technology Fund under Grant
2016JQ0004.

\end{document}